\documentclass[11pt]{article}
\usepackage[margin=1in]{geometry}
\usepackage{amssymb,amsmath,amsthm}
\usepackage{verbatim,hyperref}

\newtheorem{theorem}{Theorem}
\newtheorem{lemma}{Lemma}

\newtheorem{corollary}{Corollary}

\newtheorem{claim}{Claim}
\newtheorem{proposition}{Proposition}

\theoremstyle{definition}
\newtheorem{definition}{Definition}

\DeclareMathOperator{\alg}{ALG}

\title{Prophet Secretary: Surpassing the $1-1/e$ Barrier}
\author{
Yossi Azar\\
Tel Aviv University\\
azar@tau.ac.il
\and
Ashish Chiplunkar\thanks{This work was done when the author was a post-doc at Tel Aviv University.}\\
EPFL\\
ashish.chiplunkar@gmail.com
\and
Haim Kaplan\\
Tel Aviv University\\
haimk@post.tau.ac.il
}

\date{}

\begin{document}
\maketitle

\begin{abstract}
In the Prophet Secretary problem, samples from a known set of probability distributions arrive one by one in a uniformly random order, and an algorithm must irrevocably pick one of the samples as soon as it arrives. The goal is to maximize the expected value of the sample picked relative to the expected maximum of the distributions. This is one of the most simple and fundamental problems in online decision making that models the process selling one item to a sequence of costumers. For a closely related problem called the Prophet Inequality where the order of the random variables is adversarial, it is known that one can achieve in expectation $1/2$ of the expected maximum, and no better ratio is possible. For the Prophet Secretary problem, that is, when the variables arrive in a random order, Esfandiari et al.\ \cite{EsfandiariHLM_ESA15} showed that one can actually get $1-1/e$ of the maximum. The $1-1/e$ bound was recently extended to more general settings \cite{EhsaniHKS_arXiv17}. Given these results, one might be tempted to believe that $1-1/e$ is the correct bound. We show that this is not the case by providing an algorithm for the Prophet Secretary problem that beats the $1-1/e$ bound and achieves $1-1/e+1/400$ of the optimum value. We also prove a hardness result on the performance of algorithms under a natural restriction which we call deterministic distribution-insensitivity.
\end{abstract}


\section{Introduction}

The Prophet Inequality problem of Krengel and Sucheston \cite{KrengelS_BAMS77,KrengelS_APRT78} is one of the cornerstones in optimal stopping theory. In this problem, probability distributions of a sequence $X_1,\ldots,X_n$ of independent non-negative random variables is given to an algorithm. Thereafter, samples $x_1\sim X_1,\ldots,x_n\sim X_n$ are revealed to the algorithm one by one, and the algorithm is required to pick (at most) one of the $x_i$'s irrevocably as soon as it is revealed. The goal of the algorithm is to maximize its profit, that is, the value of the sample that it picks. It is easy to see that we may only look for \textit{threshold algorithms}, that is, algorithms which come up with a threshold $\alpha_i$ in each round $i$, and accept $x_i$ if and only if $x_i\geq\alpha_i$.

The primary motivation behind studying the Prophet Inequality and its related problems comes from mechanism design. Consider a single-item auction where customers arrive online, and the item is to be sold irrevocably to one of the customers. Suppose that the probability distribution of the valuation of each customer is known in advance. A threshold algorithm naturally translates to a \textit{posted price mechanism}, which is inherently \textit{truthful}, and the goal of maximizing the sample picked translates to maximizing the \textit{social welfare}.

Indeed, if the algorithm knew all the samples $x_1,\ldots,x_n$ in advance, it would pick $\max_ix_i$, resulting in an expected profit of $\mathbb{E}[\max_iX_i]$. The algorithm's performance is, therefore, compared against the benchmark of $\mathbb{E}[\max_iX_i]$. An algorithm is said to have a \textit{competitive ratio} of $c$ if its expected profit is at least $c\cdot\mathbb{E}[\max_iX_i]$, on any sequence $X_1,\ldots,X_n$ of random variables. The classical results of Krengel and Sucheston \cite{KrengelS_BAMS77,KrengelS_APRT78} state that $(1/2)$-competitive algorithms exist, and no algorithm can have a better competitive ratio. One such $(1/2)$-competitive algorithm is the following simple algorithm by Kleinberg and Weinberg \cite{KleinbergW_STOC12}: compute $T=\mathbb{E}[\max_iX_i]/2$, and pick the first $x_i$ which exceeds $T$, if one exists.

Another well known problem in the domain of optimal stopping is the Secretary problem \cite{Dynkin_SMD63,Ferguson_SS89}, where a set of values $\{x_1,\ldots x_n\}$ is chosen adversarially, and revealed to an algorithm in a random order. As before, an algorithm has to pick one of the $x_i$'s as soon as it is revealed. Inspired from this problem, Esfandiari, Hajiaghayi, Liaghat, and Monemizadeh \cite{EsfandiariHLM_ESA15} defined a natural variant of the Prophet Inequality problem called the Prophet Secretary problem. Here, the set $\{X_1,\ldots,X_n\}$ of random variables is adversarial and known to the algorithm in advance, whereas the samples $\{x_i\sim X_i\}$ are revealed in a uniformly random order. Since the order is no longer adversarial, it is natural to expect that algorithms with competitive ratio larger than $1/2$ will exist. Esfandiari et al.\ gave an algorithm which does the following. The algorithm chooses a sequence $1>\alpha_1>\cdots>\alpha_n>0$ of thresholds determined completely by $n$ and independent of the set of distributions. Suppose $\sigma\in\mathcal{S}_n$ is the random order in which the samples are revealed, that is, the $k^{\text{\tiny{th}}}$ sample to be revealed is $x_{\sigma(k)}$. Then the algorithm picks $x_{\sigma(k)}$ for the smallest $k$ such that $x_{\sigma(k)}\geq\alpha_k\cdot\mathbb{E}[\max_i X_i]$, if such an index $k$ exists, otherwise it picks nothing. Esfandiari et al.\ proved that, for a suitable choice of thresholds $\alpha_1,\ldots,\alpha_n$, the algorithm is $(1-1/e)$-competitive.

As an impossibility result, Esfandiari et al.\ \cite{EsfandiariHLM_ESA15} also proved that no algorithm for the prophet secretary problem can have a competitive ratio better than $3/4$. However, observe that the competitive ratio in the particular case where $X_1,\ldots,X_n$ are identical and independent is trivially an upper bound on the competitive ratio of the Prophet Secretary problem. Moreover, if $X_1,\ldots,X_n$ are identical, then the Prophet Inequality and the Prophet Secretary problems are equivalent. Hill and Kertz \cite{HillK_AP82} had already proved that the competitive ratio for IID Prophet Inequality, which we call $c_{\text{iid}}$, is at most $(1.341)^{-1}\approx0.746$. No upper bound better than $c_{\text{iid}}$ is known for the Prophet Secretary problem.

\subsection{Our results}

Our main result is that the (algorithmic) lower bound of $1-1/e$ on the competitive ratio of Prophet Secretary by Esfandiari et al.\ \cite{EsfandiariHLM_ESA15} is not tight. By introducing new techniques in addition to those of Esfandiari et al., we give an algorithm for the Prophet Secretary problem that has a competitive ratio better than $1-1/e$. (It is noteworthy that, in contrast to Abolhassani et al.\ \cite{AbolhassaniEEHK_STOC17} who beat the $1-1/e$ bound only by allowing the algorithm to choose the arrival order of random variables, we beat the bound for random arrival order.)

\begin{theorem}\label{thm_alg}
There is an algorithm for the Prophet Secretary problem with competitive ratio larger than $1-1/e+1/400$.
\end{theorem}

Observe that the algorithms for the Prophet Inequality problem as well as the Prophet Secretary problem which we stated previously are both simple in the following sense. Both algorithms are actually oblivious to the probability distributions of $X_1,\ldots,X_n$, they only need to know $\mathbb{E}[\max_iX_i]$, which they compete against. Moreover, both algorithms choose their threshold(s) deterministically. We call such algorithms \textit{deterministic distribution-insensitive algorithms}. Our second result is a hardness result for deterministic distribution-insensitive algorithms.

\begin{theorem}\label{thm_hardness}
Deterministic distribution-insensitive algorithms for the Prophet Secretary problem cannot have a competitive ratio larger than $11/15\approx0.733$.
\end{theorem}

This improves the upper bound of $3/4$ by Esfandiari et al.\ \cite{EsfandiariHLM_ESA15}, as well as the better upper bound of $c_{iid}\approx0.746$ due to Hill and Kertz \cite{HillK_AP82}, for deterministic distribution-insensitive algorithms.

\subsection{Our techniques}

We give here the high-level ideas behind the design of our algorithm for the Prophet Secretary problem. Recall that $\mathbb{E}[\max_i X_i]=\int_{0}^{\infty}\Pr[\max_i X_i\geq x]dx$, and imagine that every interval $I\subseteq\mathbb{R}$ contributes the value $\int_{x\in I}\Pr[\max_i X_i\geq x]dx$ to $\mathbb{E}[\max_i X_i]$. We divide the Prophet Secretary instances into three categories as follows. Loosely speaking, the first category contains instances in which the contribution of the interval $[0,1-1/e]$ to $\mathbb{E}[\max_i X_i]$ is small. The second category contains instances in which, in expectation, more than one $X_i$'s exceed a certain threshold. For each of these categories, we strengthen Esfandiari et al.'s algorithm to achieve a better performance.

The third category is most interesting and includes all the remaining instances. For these instances we prove that one of the $X_i$'s (say $X_1$) is larger than all the rest with high probability, and also has a sufficiently large expectation. For these instances, our algorithm sets the same threshold for all samples. We prove that with high probability, the algorithm does not pick a sample before it encounters $X_1$. As a consequence, it extracts most of the expected value of $X_1$ as its profit. Moreover, the algorithm encounters, on an average, half of the other $X_i$'s before it sees $X_1$, because the samples arrive in a uniformly random order. Thus, even in the unlikely event that one of the $X_i$'s ($i>1$) exceeds the threshold, the algorithm extracts its value with probability close to $1/2$. (This is necessary because $\mathbb{E}[X_1]$, although large, is not guaranteed to be sufficient by itself for the algorithm to achieve its targeted competitive ratio.)

\subsection{Related Work}

\noindent{\textbf{Prophet Inequalities (worst case arrival order):}} Ever since the seminal work of Krengel and Sucheston \cite{KrengelS_BAMS77,KrengelS_APRT78}, the Prophet Inequality problem has been studied in a variety of settings. One of the most natural variants is perhaps the multiple choice Prophet Inequality, where the algorithm is required to pick at most $k$ of the $x_i$'s ($k>1$), and has to compete against the expectation of the sum of the $k$ largest random variables. Alaei \cite{Alaei_FOCS11} gave an algorithm for this problem with competitive ratio $1-O(k^{-1/2})$, which is known to be asymptotically optimal. Kleinberg and Weinberg \cite{KleinbergW_STOC12} considered the Matroid Prophet Inequality problem, where the feasible subsets of random variables are independent sets of a given matroid. They gave a $(1/2)$-competitive algorithm even for this more general problem than the classical Prophet Inequality problem. Going beyond this, Prophet Inequalities where the feasibility constraint is an intersection of matroids \cite{KleinbergW_STOC12}, or even an  arbitrary downward-closed set system \cite{Rubinstein_STOC16,RubinsteinS_SODA17}, have been studied.

\noindent{\textbf{IID Prophet Inequalities:}} The case of identical distributions is particularly interesting because in this setting, the Prophet Inequality and the Prophet Secretary problems coincide. Already more than two decades ago, Hill and Kertz \cite{HillK_AP82} gave an implicit characterization of $c_{iid}$, the competitive ratio for identical distributions, by analyzing the natural dynamic programming algorithm which performs optimally on every instance. However, they could only find the numerical value of $c_{iid}(n)$, the competitive ratio when we have $n$ independent identical random variables, for all $n\leq 10000$. They conjectured that $c_{iid}(n)$ decreases as $n$ increases, and proved that for $n=10000$, the competitive ratio is approximately $(1.341)^{-1}\approx0.746$. For $n>10000$, they could only prove a lower bound of $1-1/e$ on $c_{\text{iid}}(n)$. Only recently, Abolhassani et al.\ \cite{AbolhassaniEEHK_STOC17} improved the lower bound for large $n$ to $0.738$ by giving a new algorithm. Shortly thereafter, Correa, Foncea, Hoeksma, Oosterwijk, and Vredeveld \cite{CorreaFHOV_EC17} took an entirely new approach to the problem and came up with yet another algorithm. They proved that its competitive ratio is at least $c_{\text{cfhov}}>0.745$, a constant which they defined implicitly. Surprisingly, Hill and Kertz's implicit characterization of $c_{iid}$ and Correa et al's implicit definition of $c_{\text{cfhov}}$, although quite different from their appearance, are actually the same. This means that $c_{iid}=c_{\text{cfhov}}\in[0.745,0.746]$, and that the algorithm of Correa et al.\ \cite{CorreaFHOV_EC17}, in fact, achieves the optimal competitive ratio. Correa et al.\ have made this remark already, and we give the explicit proof in Appendix \ref{app_iid}.

\noindent{\textbf{Best case arrival order:}} A natural variant of the Prophet Inequality and the Prophet Secretary problems is the one where the order in which the samples are revealed is chosen by the algorithm (the \text{best order} case). Yan \cite{Yan_SODA11} gave a $(1-1/e)$-competitive algorithm for the best case arrival order problem. Esfandiari et al.\ \cite{EsfandiariHLM_ESA15} proved that this bound holds even if the arrival order is random, as stated earlier, which suggests that the best case competitive ratio might be substantially better. Pursuing this line of research, Abolhassani et al.\ \cite{AbolhassaniEEHK_STOC17} proved that as long as there are sufficiently many (independent) copies of each random variable, an algorithm can come up with an appropriate order for which the competitive ratio is arbitrarily close to $c_{iid}$.\footnote{In fact, this claim holds even in the random order case, where $O(\log n)$ multiplicity is sufficient. If the algorithm is allowed to choose the order, even $O(1)$ multiplicity is sufficient.} In the absence of the multiplicity assumption, it is not known whether we can attain a competitive ratio for the best order case which is better than the random order case.

\noindent{\textbf{Matroid Prophet Secretary:}} Very recently, Ehsani, Hajiaghayi, Kesselheim, and Singla \cite{EhsaniHKS_arXiv17} have considered the generalization of Prophet Secretary to matroids, analogous to Kleinberg and Weinberg's generalization of Prophet Inequality to matroids. Ehsani et al.\ give an algorithm that achieves the $1-1/e$ bound even with matroid feasibility constraint.

\noindent{\textbf{Posted price mechanism design:}} As remarked earlier, algorithms for problems like the Prophet Inequality and Prophet Secretary correspond to posted price mechanisms for approximately maximizing social welfare. A parallel line of work has been to design posted price mechanisms under similar settings for approximate revenue maximization, taking as benchmark the revenue obtained by Myerson's mechanism. For more information on posted price mechanisms for online arrival of buyers, see \cite{HajiaghayiKS_AAAI07,ChawlaHMS_STOC10,Alaei_FOCS11,CorreaFHOV_EC17}, and the references therein.

\subsection{Organization of the paper}

Towards proving Theorem \ref{thm_alg}, we introduce some notation in Section \ref{sec_prelim}. Since we reuse some part of Esfandiari et al.'s proof, we give an overview of that proof in Section \ref{sec_old}. We present our algorithm in Section \ref{sec_alg}, and its analysis in the subsequent subsections, thereby proving Theorem \ref{thm_alg}. We devote Section \ref{sec_hardness} to prove Theorem \ref{thm_hardness}, where we give an adversarial strategy to defeat every distribution-insensitive algorithm.

\section{Preliminaries}\label{sec_prelim}

Let $X_1,\ldots,X_n$ be independent non-negative random variables and let $X=\max_iX_i$. By scaling the random variables appropriately, we assume $\mathbb{E}[X]=1$ without loss of generality throughout this paper.

In the \textit{Prophet Secretary} problem, the following interaction takes place between an adversary ADV, an algorithm ALG, and a third player, say RAND, responsible for all the randomness. (We abuse notation and also denote the expected profit of the algorithm by $\alg$.)
\begin{enumerate}
\item ADV declares a set $X_1,\ldots,X_n$ of independent non-negative random variables such that $\mathbb{E}[\max_i X_i]=1$. (More specifically, ADV reveals the distribution functions of the random variables.)
\item RAND samples a permutation $\sigma\in\mathcal{S}_n$ over the indices $\{1,\ldots,n\}$ uniformly at random and keeps it private.
\item For $k=1$ to $n$,
\begin{enumerate}
\item ALG declares a threshold $\alpha_k$.
\item RAND samples $x_{\sigma(k)}\sim X_{\sigma(k)}$. If $x_{\sigma(k)}\geq\alpha_k$, then the ALG is given the profit $x_{\sigma(k)}$ and the loop breaks.
\end{enumerate}
\end{enumerate}
Note that in the above setting, we assume that the algorithm does not know $\sigma$, and it does not know $x_{\sigma(k)}$ even if it rejects $x_{\sigma(k)}$. 

In the more restricted setting of the \textit{Distribution-Insensitive Prophet Secretary} problem, ADV, instead of making the random variables $X_1,\ldots,X_n$ public, sends them to RAND only. Thus, the algorithm only knows that $\mathbb{E}[\max_iX_i]$ is one. The thresholds $\alpha_1,\ldots,\alpha_n$ chosen by the algorithm are thus determined solely by $n$, and are oblivious to the probability distributions of $X_1,\ldots,X_n$. Moreover, if an algorithm chooses $\alpha_1,\ldots,\alpha_n$ deterministically, we call it a \textit{deterministic distribution-insensitive} algorithm. 
The $(1-1/e)$-competitive algorithm of Esfandiari et al.\ \cite{EsfandiariHLM_ESA15} is a deterministic distribution-insensitive algorithm. In contrast, the algorithm that we design is sensitive to the distributions of $X_1,\ldots,X_n$.

\section{Overview of the old analysis}\label{sec_old}

We briefly restate those parts of the analysis by Esfandiari et al.\ which we require in our proof. Let $\alpha_1,\ldots,\alpha_n$ be the thresholds chosen by an algorithm. Let the random variable $z_k$ denote the profit obtained by the algorithm from the $k^{\text{\tiny{th}}}$ round. That is, $z_k$ is equal to $x_{\sigma(k)}$ if the algorithm picked $x_{\sigma(k)}$, the $k^{\text{\tiny{th}}}$ sample; otherwise $z_k$ is zero. Then
\begin{equation}\label{eqn_alg_esa}
\alg=\sum_{k=1}^n\mathbb{E}[z_k]=\sum_{k=1}^n\int_0^{\infty}\Pr[z_k\geq x]dx=\sum_{k=1}^n\int_0^{\alpha_k}\Pr[z_k\geq x]dx+\sum_{k=1}^n\int_{\alpha_k}^{\infty}\Pr[z_k\geq x]dx\text{.}
\end{equation}
The two terms in the last expression above are bounded from below as follows. Let $\theta(k)$ denote the probability that the algorithm does not choose a value from the first $k$ samples. Let $\alpha_{n+1}=0$. Then we have,

\begin{proposition}[Lemma 10 of \cite{EsfandiariHLM_ESA15}]\label{prop_0}
\[\sum_{k=1}^n\int_0^{\alpha_k}\Pr[z_k\geq x]dx=\sum_{k=1}^n(1-\theta(k))(\alpha_k-\alpha_{k+1})=\alpha_1-\sum_{k=1}^n\theta(k)(\alpha_k-\alpha_{k+1})\]
\end{proposition}

\begin{proposition}[Lemma 11 of \cite{EsfandiariHLM_ESA15}]\label{prop_infty}
$\Pr[z_k\geq x]\geq\frac{\theta(k)}{n}\Pr[X\geq x]$ for $x\geq\alpha_k$, and
\[\int_{\alpha_k}^{\infty}\Pr[z_k\geq x]dx\geq\frac{\theta(k)}{n}\left(1-\int_0^{\alpha_k}\Pr[X\geq x]dx\right)\geq\frac{\theta(k)}{n}\cdot(1-\alpha_k)\text{.}\]
\end{proposition}

Substituting these bounds into (\ref{eqn_alg_esa}), we get a lower bound on $\alg$ which depends on the $\theta(k)$'s. Esfandiari et al.\ carefully choose the thresholds $\alpha_1,\ldots,\alpha_n$ in such a way that this dependence is eliminated, and $\alpha_1\approx1-1/e$ is a lower bound on $\alg$, and hence, on the competitive ratio.

\section{The improved Algorithm}\label{sec_alg}

The main intuition behind our Prophet Secretary algorithm is the following. Suppose that the set $\{X_1,\ldots,X_n\}$ of random variables is such that for some $a$ and $c$, the average value of $\Pr[X\geq x]$ over $x\in[0,a]$ is at most $c<1$. Then we can use this to strengthen Proposition \ref{prop_infty}. Next, suppose that for some $T$, $b$, and $d$, $\sum_i\Pr[X_i\geq x]$ is at least $b$ and $\int_T^{\infty}\Pr[X\geq x]dx=d$. Then we also exploit this to strengthen Proposition \ref{prop_infty}. In either case, we improve the competitive ratio. However, there exist instances in which none of these improvements is possible. We prove that such instances must have one of the $X_i$'s that is bounded away from zero with high probability, while the rest are close to zero with high probability. In this case, we prove that using a uniform threshold for all samples suffices.

Our improved algorithm for the Prophet Secretary problem has the following four parameters: $a\in[0,1-1/e]$, $b>1$, $c\in[0,1]$, and $d\in[0,1]$, chosen such that
\begin{equation}\label{eqn_acd}
ac+d>1\text{.}
\end{equation}
These are absolute constants independent of $n$ and the distributions of the random variables, and their values will be fixed later. Let $T$ be such that $\int_T^{\infty}\Pr[X\geq x]dx=d$, or equivalently, $\int_0^T\Pr[X\geq x]dx=1-d$, because we assumed that $\mathbb{E}[X]=\int_0^{\infty}\Pr[X\geq x]dx=1$. Since $\Pr[X\geq x]\leq 1$, this implies the following lower bound on $T$.
\begin{equation}\label{eqn_T_lb}
T\geq\int_0^T\Pr[X\geq x]dx=1-d\text{.}
\end{equation}
The thresholds $\alpha_1,\ldots,\alpha_n$ of the algorithm are determined as follows.
\begin{itemize}
\item \textbf{Case 1:} If $\int_0^a\Pr[X\geq x]dx\leq c\cdot a$, then
\[\alpha_k=\begin{cases}
1-e^{(k-1)/n-1} & \text{if }\frac{k-1}{n}>1+\ln(1-a)\text{,}\\
\frac{1}{c}-\left(\frac{1}{c}-a\right)\left(\frac{e^{(k-1)/n-1}}{(1-a)}\right)^c & \text{otherwise.}
\end{cases}\]
\item \textbf{Case 2:} else if $\sum_{i=1}^n\Pr[X_i\geq T]\geq b$, then
\[\alpha_k=\begin{cases}
(1-e^{b((k-1)/n-1)})\cdot\frac{b+d-bd}{b} & \text{if }\frac{k-1}{n}>1-\frac{1}{b}\ln\frac{b+d-bd}{d}\text{,}\\
1-d\cdot\left(\frac{b+d-bd}{d}\right)^{1/b}\cdot e^{(k-1)/n-1} & \text{otherwise.}
\end{cases}\]
\item \textbf{Case 3:} else $\alpha_1=\cdots=\alpha_n=T$.
\end{itemize}

Our analyses of Case 1 and Case 2 reuse some parts of the analysis by \cite{EsfandiariHLM_ESA15}, whereas the analysis of Case 3 uses novel approach. We present the analyses of the three cases in the upcoming subsections, and then choose our parameters to get the final bound on the competitive ratio.


\subsection{Analysis of Case 1}\label{subsec_1}

Suppose that the input set of distributions falls into Case 1, that is, $\int_0^a\Pr[X\geq x]dx\leq c\cdot a$. Using this upper bound on $\int_0^a\Pr[X\geq x]dx$, we get the following strengthened version of Proposition \ref{prop_infty}.

\begin{lemma}\label{lem_case1}
If $\frac{k-1}{n}\leq1+\ln(1-a)$ then $\int_{\alpha_k}^{\infty}\Pr[z_k\geq x]dx\geq\frac{\theta(k)}{n}\cdot(1-c\alpha_k)$.
\end{lemma}

\begin{proof}
From the later part of Proposition \ref{prop_infty}, we have,
\begin{equation}\label{eqn_reuse1}
\int_{\alpha_k}^{\infty}\Pr[z_k\geq x]dx\geq\frac{\theta(k)}{n}\left(1-\int_0^{\alpha_k}\Pr[X\geq x]dx\right)\text{.}
\end{equation} 
Suppose $(k-1)/n\leq1+\ln(1-a)$. From the definition of $\alpha_k$ in Case 1 of the algorithm, it is clear that $\alpha_k\geq a$. Thus, for any $x\geq0$, $ax/\alpha_k\leq x$, and hence, $\Pr[X\geq x]\leq\Pr[X\geq ax/\alpha_k]$. Therefore,
\[\int_0^{\alpha_k}\Pr[X\geq x]dx\leq\int_0^{\alpha_k}\Pr\left[X\geq\frac{ax}{\alpha_k}\right]dx=\frac{\alpha_k}{a}\int_0^a\Pr[X\geq y]dy\leq\frac{\alpha_k}{a}\cdot ca=c\alpha_k\]
where we get the first equality with the substitution $y=ax/\alpha_k$ and the second last inequality follows from the assumption that we apply Case 1. Substituting this upper bound on $\int_0^{\alpha_k}\Pr[X\geq x]dx$ in (\ref{eqn_reuse1}), the claim follows.
\end{proof}

Using Lemma \ref{lem_case1} to bound $\int_{\alpha_k}^{\infty}\Pr[z_k\geq x]dx$ for positions $k$ such that $(k-1)/n\leq1+\ln(1-a)$, and Proposition \ref{prop_infty} for positions $k$ such that $(k-1)/n>1+\ln(1-a)$, we get a stronger lower bound on $\sum_{k=1}^n\int_{\alpha_k}^{\infty}\Pr[z_k\geq x]dx$ than what is obtained by using Proposition \ref{prop_infty} alone. Adding to this the bound on $\sum_{k=1}^n\int_0^{\alpha_k}\Pr[z_k\geq x]dx$ given by Proposition \ref{prop_0}, we get a lower bound on $\alg$ that depends on the $\theta(k)$'s. Our choice of the values of $\alpha_k$'s reduces this dependence to an $O(n^{-1})$ error term, and we get the following bound on the competitive ratio, whose proof is deferred to Appendix \ref{app_1}.

\begin{theorem}\label{thm_case1}
Suppose that the input set of distributions falls into Case 1 of the algorithm. Then the algorithm's expected profit is at least $\frac{1}{c}-\left(\frac{1}{c}-a\right)(e(1-a))^{-c}-\frac{\gamma_1}{n}$, for an absolute constant $\gamma_1$.
\end{theorem}

\subsection{Analysis of Case 2}\label{subsec_2}

Suppose that the input set of distributions falls into Case 2, that is, $\sum_{i=1}^n\Pr[X_i\geq T]\geq b$. Then $\sum_{i=1}^n\Pr[X_i\geq x]\geq b$ for all $x\leq T$, because each $\Pr[X_i\geq x]$ can only increase as $x$ decreases. This gives us the following strengthened version of Proposition \ref{prop_infty}.

\begin{lemma}\label{lem_case2}
If $\frac{k-1}{n}\geq1-\frac{1}{b}\ln\frac{b+d-bd}{d}$ then $\int_{\alpha_k}^{\infty}\Pr[z_k\geq x]dx\geq\frac{\theta(k)}{n}\cdot(b+d-bd-b\alpha_k)$.
\end{lemma}

\begin{proof}
Suppose $\frac{k-1}{n}\geq1-\frac{1}{b}\ln\frac{b+d-bd}{d}$. From the definition of $\alpha_k$ in Case 2 of the algorithm, it is clear that $\alpha_k\leq 1-d$. Thus, from (\ref{eqn_T_lb}), we have $\alpha_k\leq 1-d\leq T$. For $x\in[\alpha_k,T]$, we have
\begin{equation}\label{eqn_reuse2}
\Pr[z_k\geq x]\geq\frac{\theta(k)}{n}\sum_{i=1}^n\Pr[X_i\geq x]\geq\frac{\theta(k)}{n}\sum_{i=1}^n\Pr[X_i\geq T]\geq b\cdot\frac{\theta(k)}{n}\text{.}
\end{equation} 
Thus,
\begin{eqnarray*}
\int_{\alpha_k}^{\infty}\Pr[z_k\geq x]dx & = & \int_{\alpha_k}^T\Pr[z_k\geq x]dx+\int_T^{\infty}\Pr[z_k\geq x]dx\\
 & \geq & \int_{\alpha_k}^Tb\cdot\frac{\theta(k)}{n}dx+\int_T^{\infty}\frac{\theta(k)}{n}\Pr[X\geq x]dx\\
 & \geq & \frac{\theta(k)}{n}\cdot b\cdot(1-d-\alpha_k)+\frac{\theta(k)}{n}\cdot d\\
 & = & \frac{\theta(k)}{n}\cdot(b+d-bd-b\alpha_k)\text{.}
\end{eqnarray*}
Here, in the first inequality, we used Equation (\ref{eqn_reuse2}) for the first term and the earlier part of Proposition \ref{prop_infty} for the second term. In the second inequality, we used the lower bound (\ref{eqn_T_lb}) on $T$ for the first term and the definition of $T$ for the second term.
\end{proof}

Analogous to Case 1, using Lemma \ref{lem_case2} to bound $\int_{\alpha_k}^{\infty}\Pr[z_k\geq x]dx$ for positions $k$ such that $(k-1)/n\geq1-\frac{1}{b}\ln\frac{b+d-bd}{d}$, and Proposition \ref{prop_infty} for positions $k$ such that $(k-1)/n<1-\frac{1}{b}\ln\frac{b+d-bd}{d}$, we get better bound on $\sum_{k=1}^n\int_{\alpha_k}^{\infty}\Pr[z_k\geq x]dx$. Using this bound and Proposition \ref{prop_0}, we again get a bound $\alg$ which depends on the $\theta(k)$'s. As before, the choice of $\alpha_k$'s reduces this dependence to an $O(n^{-1})$ error term, and we get the following bound on the competitive ratio, whose proof is deferred to Appendix \ref{app_2}

\begin{theorem}\label{thm_case2}
Suppose that the input set of distributions falls into Case 2 of the algorithm. Then the algorithm's expected profit is at least $1-\frac{d}{e}\cdot\left(\frac{b+d-bd}{d}\right)^{1/b}-\frac{\gamma_2}{n}$, for an absolute constant $\gamma_2$.
\end{theorem}

\subsection{Analysis of Case 3}\label{subsec_3}

Suppose the set $\{X_1,\ldots,X_n\}$ of random variables falls into Case 3 of the algorithm. Then the algorithm chooses $\alpha_1=\cdots=\alpha_n=T$, where $\int_T^{\infty}\Pr[X\geq x]dx=d$ and $\int_0^T\Pr[X\geq x]dx=1-d$. Consider the event that the algorithm does not pick any of the samples, and recall that its probability is $\theta(n)$, by definition. This event happens if and only if all the samples are less than $T$, or equivalently, the maximum of the samples is less than $T$. Our first lemma states that this event is not too likely.

\begin{lemma}\label{lem_thetan}
Define $h=\frac{ca-1+d}{a-1+d}$. Then the algorithm chooses some sample with probability at least $h$, that is,
\[1-\theta(n)=\Pr[\exists i\text{ }X_i\geq T]=\Pr[X\geq T]>h\text{.}\]
\end{lemma}

\begin{proof}
Since we are in Case 3, we have $\int_0^a\Pr[X\geq x]dx>ca$. Recall Equation (\ref{eqn_acd}), which stated that we choose $a$, $c$, and $d$ such that $ac+d>1$. Thus,
\[\int_0^T\Pr[X\geq x]dx=1-d<ac<\int_0^a\Pr[X\geq x]dx\text{.}\]
This implies $T<a$. For any $x\in[T,a]$, $\Pr[X\geq x]\leq\Pr[X\geq T]$. Thus,
\[1-d+(a-T)\Pr[X\geq T]\geq\int_{0}^T\Pr[X\geq x]dx+\int_T^a\Pr[X\geq x]dx=\int_0^a\Pr[X\geq x]dx>ca\text{.}\]
Therefore, 
\[1-\theta(n)=\Pr[X\geq T]>\frac{ca-1+d}{a-T}\geq\frac{ca-1+d}{a-1+d}=h\]
where we used the lower bound on $T$ given by (\ref{eqn_T_lb}) for the second inequality.
\end{proof}

Next, we bound $\sum_{k=1}^n\int_0^{\alpha_k}\Pr[z_k\geq x]dx$ and $\sum_{k=1}^n\int_{\alpha_k}^{\infty}\Pr[z_k\geq x]dx$ from below. Substituting these lower bounds in Equation (\ref{eqn_alg_esa}), we get a lower bound on the algorithm's profit. From Proposition \ref{prop_0} and recalling that $\alpha_{n+1}=0$, we have,
\begin{equation}\label{eqn_term1_lb}
\sum_{k=1}^n\int_0^{\alpha_k}\Pr[z_k\geq x]dx=\alpha_1-\sum_{k=1}^n\theta(k)(\alpha_k-\alpha_{k+1})=T-\theta(n)\cdot T=(1-\theta(n))T\geq h(1-d)
\end{equation}
where we used Lemma \ref{lem_thetan} and Equation (\ref{eqn_T_lb}) for the last inequality. 

Reusing some notation and arguments from Esfandiari et al.\ \cite{EsfandiariHLM_ESA15}, let $q_{-i}(k)$ denote the probability that the algorithm rejects the first $k$ samples, given that none of them came from $X_i$. Then $\Pr[z_k\geq x]=\sum_{i=1}^n\Pr[\sigma(k)=i]\cdot\Pr[z_k\geq x\text{ }|\text{ }\sigma(k)=i]=\frac{1}{n}\sum_{i=1}^n\Pr[z_k\geq x\text{ }|\text{ }\sigma(k)=i]$. Suppose $x\geq\alpha_k=T$. Conditioned on $\sigma(k)=i$, the event $z_k\geq x$ happens if and only if the following two independent events happen: the algorithm rejects the first $k-1$ samples, and $X_i\geq x$. Thus, $\Pr[z_k\geq x\text{ }|\text{ }\sigma(k)=i]=q_{-i}(k-1)\cdot\Pr[X_i\geq x]$, and consequently, $\Pr[z_k\geq x]=\frac{1}{n}\sum_{i=1}^nq_{-i}(k-1)\cdot\Pr[X_i\geq x]$. Therefore,
\[\sum_{k=1}^n\int_{x=\alpha_k}^{\infty}\Pr[z_k\geq x]dx=\sum_{k=1}^n\int_{x=T}^{\infty}\frac{1}{n}\sum_{i=1}^nq_{-i}(k-1)\cdot\Pr[X_i\geq x]dx\text{.}\]
Interchanging the order of summations and integration, we get,
\[\sum_{k=1}^n\int_{x=\alpha_k}^{\infty}\Pr[z_k\geq x]dx=\sum_{i=1}^n\left(\frac{1}{n}\sum_{k=1}^nq_{-i}(k-1)\right)\cdot\left(\int_{x=T}^{\infty}\Pr[X_i\geq x]dx\right)\text{.}\]
Define $\mu_i=\int_{x=T}^{\infty}\Pr[X_i\geq x]dx$, the quantity in the second pair of parentheses above. Thus,
\[\sum_{k=1}^n\int_{x=\alpha_k}^{\infty}\Pr[z_k\geq x]dx=\sum_{i=1}^n\left(\frac{1}{n}\sum_{k=1}^nq_{-i}(k-1)\right)\cdot\mu_i\text{.}\]
Define $E_i$ to be the event that the algorithm encounters $X_i$, that is, it does not choose a sample before it sees $x_i$. Then observe that $\Pr[E_i|\sigma(k)=i]=q_{-i}(k-1)$, and thus, $\Pr[E_i]=\sum_{k=1}^n\Pr[\sigma(k)=i]\cdot\Pr[E_i|\sigma(k)=i]=\frac{1}{n}\sum_{k=1}^nq_{-i}(k-1)$, the expression in the parentheses above. Therefore,
\begin{equation}\label{eqn_term2_3}
\sum_{k=1}^n\int_{x=\alpha_k}^{\infty}\Pr[z_k\geq x]dx=\sum_{i=1}^n\Pr[E_i]\cdot\mu_i\text{.}
\end{equation}

In order to lower bound the above, we need a crucial lemma, which states that there is one \textit{prominent} random variable among $\{X_1,\ldots,X_n\}$ which is larger than $T$ with a large probability, whereas the others are unlikely to exceed $T$. For the rest of this section, we assume, without loss of generality, that $\Pr[X_1\geq T]=\max_i\Pr[X_i\geq T]$ (that is, the \textit{prominent} random variable is $X_1$).

\begin{lemma}\label{lem_x1}
Define $g\in(0,1]$ to be the unique\footnote{Observe that $\zeta:(0,1]\longrightarrow[0,1/e)$ defined as $\zeta(z)=(1-z)^{1/z}$ is a monotonically decreasing function with $\zeta(1)=0$ and $\lim_{z\rightarrow 0}\zeta(z)=1/e$. We will ensure that $0\leq(1-h)^{1/b}<1/e$, so that $g$ exists.} number such that $(1-g)^{1/g}=(1-h)^{1/b}$. Then $\Pr[X_1\geq T]\geq g$ and $\sum_{i=2}^n\Pr[X_i\geq T]\leq b-g$.
\end{lemma}

The proof of Lemma \ref{lem_x1} relies on the following technical result.

\begin{lemma}\label{lem_optn}
Suppose $y_1,\ldots,y_n$ are such that $0\leq y_i\leq p$ for all $i$, and $\sum_{i=1}^ny_i\leq b$, where $p\in[0,1]$. Then we have $\prod_{i=1}^n(1-y_i)\geq(1-p)^{b/p}$.
\end{lemma}

\begin{proof}
Let $(y_1,\ldots,y_n)$ minimize $\prod_{i=1}^n(1-y_i)$ subject to the constraints $0\leq y_i\leq p$ for all $i$, and $\sum_{i=1}^ny_i\leq b$. If $np\leq b$, then it is easy to see that $y_1=\ldots=y_n=p$ is the optimum. Then $\prod_{i=1}^n(1-y_i)=(1-p)^n\geq(1-p)^{b/p}$, because $1-p\in[0,1]$ and $n\leq b/p$. This implies the claim.

Now suppose $np>b$. Then some $y_i$ must be less than $p$. Further, $\sum_{i=1}^ny_i$ must be equal to $b$, otherwise we can increase one of the $y_i$s, resulting in a decrease in the objective value while maintaining feasibility. If we have $0<y_i\leq y_j<p$ for some $i\neq j$, then we may choose an appropriate $\varepsilon>0$ and replace $(y_i,y_j)$ by $(y_i-\varepsilon,y_j+\varepsilon)$, thereby decreasing $\prod_{i=1}^n(1-y_i)$ while still maintaining feasibility, and contradicting the optimality of $(y_1,\ldots,y_n)$. Thus, we must have at most one $y_i$ which is in $(0,p)$. This forces that $\lfloor b/p\rfloor$ many $y_i$s are $p$, one $y_i$ is $b-p\lfloor b/p\rfloor=p(b/p-\lfloor b/p\rfloor)$, and the rest are zero. Thus,
\[\prod_{i=1}^n(1-y_i)=(1-p)^{\lfloor b/p\rfloor}(1-p(b/p-\lfloor b/p\rfloor))\geq(1-p)^{\lfloor b/p\rfloor}(1-p)^{(b/p-\lfloor b/p\rfloor)}=(1-p)^{b/p}\]
where we used the identity $1-xz\geq(1-x)^z$ for $x,z\in[0,1]$ as long as $1-x$ and $z$ are not both zero.
\end{proof}

As a consequence, Lemma \ref{lem_x1} is proved as follows.

\begin{proof}[Proof of Lemma \ref{lem_x1}]
Suppose, for contradiction, that $\Pr[X_1\geq T]=\max_i\Pr[X_i\geq T]<g$. We have $\sum_{i=1}^n\Pr[X_i\geq T]<b$, otherwise the algorithm would execute Case 2 and not Case 3. Applying Lemma \ref{lem_optn} (with $y_i=\Pr[X_i\geq T]$), we have
\[\Pr[\exists i\text{ }X_i\geq T]=1-\prod_{i=1}^n(1-\Pr[X_i\geq T])\leq1-(1-g)^{b/g}=1-(1-h)=h\text{.}\]
This contradicts Lemma \ref{lem_thetan}. Thus, $\Pr[X_1\geq T]\geq g$, which also implies $\sum_{i=2}^n\Pr[X_i\geq T]\leq b-g$, because $\sum_{i=1}^n\Pr[X_i\geq T]<b$.
\end{proof}

Lemma \ref{lem_x1} enables us to prove Lemma \ref{lem_qi}, which gives a lower bound on $\Pr[E_i]$. Lemma \ref{lem_qi} states that the \textit{prominent} random variable $X_1$ is encountered with probability close to $1$, whereas each of the others is encountered with probability almost $1/2$. An intuitive explanation for this is the following. Lemma \ref{lem_x1} states that $X_2,\ldots,X_n$ are together unlikely to cause the algorithm to stop. Therefore, the algorithm must see $X_1$ with probability close to one. Moreover, each other $X_i$ appears before $X_1$ with probability $1/2$. Given that this happens, the algorithm is unlikely to stop before it sees $X_i$. Thus, $X_i$ is seen with probability close to $1/2$.

\begin{lemma}\label{lem_qi}
$\Pr[E_1]\geq1-\frac{b-g}{2}$, and for $i>1$, we have $\Pr[E_i]\geq\frac{1}{2}\left(1-\frac{b-g}{3}\right)$.
\end{lemma}

\begin{proof}
Let $S_i=\{j\in[n]\text{ }|\text{ }\sigma^{-1}(j)<\sigma^{-1}(i)\}\subseteq[n]\setminus{i}$ be the random subset of indices $j$ whose position is before $i$ in the random permutation $\sigma$. Then
\begin{equation}\label{eqn_Ei}
\Pr[E_i\text{ }|\text{ }S_i]=\Pr[X_j<T\text{ }\forall j\in S_i]\geq1-\sum_{j\in S_i}\Pr[X_j\geq T]=1-\sum_{j\neq i}\Pr[X_j\geq T]\cdot\mathbb{I}[j\in S_i]
\end{equation}
where the inequality follows by the union bound.

First, consider the case of $i=1$. By (\ref{eqn_Ei}), we have,
\[\Pr[E_1]\geq\mathbb{E}_{S_1}\left[1-\sum_{j=2}^n\Pr[X_j\geq T]\cdot\mathbb{I}[j\in S_1]\right]=1-\sum_{j=2}^n\Pr[X_j\geq T]\cdot\Pr[j\in S_1]\text{.}\]
Every $j\neq1$ is equally likely to be before $1$ and after $1$ in the random permutation. Thus, $\Pr[j\in S_1]=1/2$. Therefore,
\[\Pr[E_1]\geq1-\frac{1}{2}\sum_{j=2}^n\Pr[X_j\geq T]\geq1-\frac{b-g}{2}\]
where we used Lemma \ref{lem_x1} for the second inequality.

Next, let $i>1$. Again, by (\ref{eqn_Ei}), we have,
\begin{eqnarray*}
\Pr[E_i\text{ }|\text{ }1\notin S_i] & \geq & \mathbb{E}_{S_i}\left[1-\sum_{j\neq i}\Pr[X_j\geq T]\cdot\mathbb{I}[j\in S_i]\text{ }\arrowvert\text{ }1\notin S_i\right]\\
 & = & 1-\sum_{j\notin\{1,i\}}\Pr[X_j\geq T]\cdot\Pr[j\in S_i\text{ }|\text{ }1\notin S_i]\text{.}
\end{eqnarray*}
Given $1\notin S_i$, that is, $1$ is after $i$, it is equally likely that $j\notin\{1,i\}$ is before $i$, between $i$ and $1$, and after $1$, in the random permutation. Thus, $\Pr[j\in S_i\text{ }|\text{ }1\notin S_i]=1/3$. Therefore,
\[\Pr[E_i\text{ }|\text{ }1\notin S_i]\geq1-\frac{1}{3}\sum_{j\notin\{1,i\}}\Pr[X_j\geq T]\geq1-\frac{1}{3}\sum_{j\neq1}\Pr[X_j\geq T]\geq1-\frac{b-g}{3}\]
where we used Lemma \ref{lem_x1} for the last inequality. Therefore,
\[\Pr[E_i]\geq\Pr[1\notin S_i]\cdot\Pr[E_i\text{ }|\text{ }1\notin S_i]\geq\frac{1}{2}\left(1-\frac{b-g}{3}\right)\]
as required.
\end{proof}

Substituting the bounds given by Lemma \ref{lem_qi} in Equation (\ref{eqn_term2_3}), we get,
\begin{eqnarray}\label{eqn_term2_lb}
\sum_{k=1}^n\int_{\alpha_k}^{\infty}\Pr[z_k\geq x]dx & \geq & \left(1-\frac{b-g}{2}\right)\cdot\mu_1+\frac{1}{2}\left(1-\frac{b-g}{3}\right)\cdot\sum_{i=2}^n\mu_i\nonumber\\
 & = & \left(\frac{1}{2}-\frac{b-g}{3}\right)\cdot\mu_1+\frac{1}{2}\left(1-\frac{b-g}{3}\right)\cdot\sum_{i=1}^n\mu_i\text{.}
\end{eqnarray}
We are thus left to bound $\mu_1$ and $\sum_{i=1}^n\mu_i$ from below.

\begin{lemma}\label{lem_mu1}
$\mu_1\geq ca-(1-d)-(a-1+d)(b-g)$.
\end{lemma}

\begin{proof}
By the union bound, we have,
\[\int_T^a\sum_{i=1}^n\Pr[X_i\geq x]dx\geq\int_T^a\Pr[X\geq x]dx=\int_0^a\Pr[X\geq x]dx-\int_0^T\Pr[X\geq x]dx\text{.}\]
Using the facts that $\int_0^a\Pr[X\geq x]dx>ca$ (else we would be in Case 1) and $\int_0^T\Pr[X\geq x]dx=1-d$ (definition of $T$), we have,
\begin{equation}\label{eqn_1n}
\int_T^a\sum_{i=1}^n\Pr[X_i\geq x]dx\geq ca-(1-d)\text{.}
\end{equation}
On the other hand, using Lemma \ref{lem_x1} and Equation (\ref{eqn_T_lb}), we also have,
\begin{equation}\label{eqn_2n}
\int_T^a\sum_{i=2}^n\Pr[X_i\geq x]dx\leq(a-T)(b-g)\leq(a-1+d)(b-g)\text{.}
\end{equation}
Subtracting Equation \ref{eqn_2n} from Equation \ref{eqn_1n}, we get,
\[\int_T^a\Pr[X_1\geq x]dx\geq ca-(1-d)-(a-1+d)(b-g)\text{.}\]
The claim follows by observing that $\mu_1=\int_T^{\infty}\Pr[X_1\geq x]dx\geq\int_T^a\Pr[X_1\geq x]dx$.
\end{proof}

Finally, $\sum_{i=1}^n\mu_i$ is easily bounded as
\begin{equation}\label{eqn_sum_mu}
\sum_{i=1}^n\mu_i=\sum_{i=1}^n\int_T^{\infty}\Pr[X_i\geq x]dx=\int_T^{\infty}\sum_{i=1}^n\Pr[X_i\geq x]dx\geq\int_T^{\infty}\Pr[X\geq x]dx=d
\end{equation}
where the inequality is due to union bound.

As a result of the above analysis, we get the following bound on the competitive ratio.

\begin{theorem}\label{thm_case3}
Suppose that the input set of distributions falls into Case 3 of the algorithm. Then the algorithm's expected profit is at least
\[h(1-d)+\left(\frac{1}{2}-\frac{b-g}{3}\right)\cdot(ca-(1-d)-(a-1+d)(b-g))+\frac{1}{2}\left(1-\frac{b-g}{3}\right)\cdot d\text{.}\]
\end{theorem}

\begin{proof}
Substituting Equation (\ref{eqn_term1_lb}) and Equation (\ref{eqn_term2_lb}) into Equation (\ref{eqn_alg_esa}), and then using the bounds given by Equation (\ref{eqn_sum_mu}) and Lemma \ref{lem_mu1}, the claim follows.
\end{proof}

\subsection{The Overall Competitive Ratio}\label{subsec_cr}

We take $a$, $c$, $d$, and $g$ as independent parameters, so that $h=\frac{ca-1+d}{a-1+d}$, and because $(1-g)^{1/g}=(1-h)^{1/b}$, we get that $b=g\ln(1-h)/\ln(1-g)$. We numerically verified that for $a=1-1.31/e$, $c=0.98$, $d=0.62$, and $g=0.88$, we get the desired competitive ratio. It is easy to check $ac+d>1$, which we promised earlier. We get $h\approx0.925$, $b\approx1.075$, and the competitive ratio lower bounds given by Theorems \ref{thm_case1}, \ref{thm_case2}, and \ref{thm_case3} are all at least $c^*-O(n^{-1})$, for some $c^*>1-1/e+1/400$. It remains to show how the $O(n^{-1})$ term can be eliminated.

\begin{lemma}\label{lem_reduction}
If there exists a $c$-competitive algorithm for the Prophet Secretary problem with $N$ random variables, then there exists a $c$-competitive algorithm for the Prophet Secretary problem with $n$ random variables for every $n<N$.
\end{lemma}

\begin{proof}
Let ALG($N$) denote the algorithm. The required algorithm ALG($n$) behaves as follows. Given $n$ random variables $X_1,\ldots,X_n$, it adds $N-n$ ghost random variables $X_{n+1},\ldots,X_N$ which take value zero deterministically. This does not change the expectation of the maximum. We now argue that ALG($n$) can essentially simulate the behavior of ALG($N$). Given a permutation $\sigma$ of $[n]$, we can generate a permutation $\sigma'$ of $[N]$ as follows. Arrange $n+1,\ldots,N$ uniformly at random in $N-n$ out of $N$ locations, then place $1,\ldots,n$ in the $n$ vacant locations in the order specified by $\sigma$. If $\sigma$ is a uniformly random permutation of $[n]$, then $\sigma'$ is a uniformly random permutation of $[N]$. ALG($n$) uses this trick as follows. First it places the ghost random variables into $N-n$ out of $N$ locations, and starts running ALG($N$). As soon as it encounters a vacant location, it asks for the next real input, and passes it on to ALG($N$). Clearly, since ALG($N$) is $c$-competitive, so is ALG($n$).

It is instructive to simplify the working of ALG($n$). Given the sequence of thresholds $\alpha_1,\ldots,\alpha_N$ chosen by ALG($N$), ALG($n$) essentially chooses a random subsequence of $n$ thresholds and uses them. Clearly, since the algorithm is $c$-competitive on an average, there exists a subsequence of $\alpha_1,\ldots,\alpha_N$ with $n$ thresholds which achieves $c$-competitiveness.
\end{proof}

\begin{corollary}
If there exists an algorithm with competitive ratio $c-O(n^{-1})$ for the Prophet Secretary problem on $n$ random variables, then for any $\varepsilon>0$ there exists a $(c-\varepsilon)$-competitive algorithm for the Prophet Secretary problem.
\end{corollary}

\begin{proof}
Given an instance of the Prophet Secretary problem on $n$ random variables, choose $N$ large enough so that the $O(N^{-1})$ term in the competitive ratio is less than $\varepsilon$, and apply Lemma \ref{lem_reduction}.
\end{proof}

As a consequence of the above corollary, we have an algorithm for the Prophet Secretary problem with competitive ratio $c^*-\varepsilon$ for every $\varepsilon>0$. Since $c^*>1-1/e+1/400$, Theorem \ref{thm_alg} follows.

\section{A Hardness Result for Distribution-Insensitive Algorithms}\label{sec_hardness}

Recall the definition of a deterministic distribution-insensitive algorithm from Section \ref{sec_prelim}: such an algorithm chooses its thresholds $\alpha_1,\ldots,\alpha_n$ deterministically merely with the knowledge of the expected value of the maximum of the random variables, which we assumed to be one. Therefore, the adversary may adapt the input set of random variables to the algorithm subject to keeping the expectation of their maximum to be one. We now give an adversarial strategy which forces an upper bound of $11/15$ on the competitive ratio of any deterministic distribution-insensitive algorithm ALG. This proves Theorem \ref{thm_hardness}.

Let $n=3$, and let $\alpha_1,\alpha_2,\alpha_3$ be the thresholds chosen by a deterministic distribution-insensitive algorithm ALG. Without loss of generality, we assume $\alpha_3=0$, because if the algorithm does not pick any of the first two samples, then it can only benefit by picking the last one, no matter how small it is. Depending on the values of $\alpha_1$ and $\alpha_2$, the adversary chooses one of the following instances as input to the algorithm, where $\varepsilon>0$ is an arbitrarily small constant.
\begin{itemize}
\item \textbf{Instance 1}: $X_1$ is $1$ deterministically; $X_2$ and $X_3$ are both $0$ deterministically.
\item \textbf{Instance 2}: $X_1$ is $1$ deterministically; $X_2$ and $X_3$ are both $\alpha_1$ deterministically.
\item \textbf{Instance 3}: $X_1$ is $(1-(1-\varepsilon)(\alpha_1-\varepsilon))/\varepsilon$ with probability $\varepsilon$ and $\alpha_1-\varepsilon$ with probability $1-\varepsilon$; $X_2$ and $X_3$ are both $\alpha_2$ deterministically. (This instance is used only if $\alpha_1>\alpha_2$.)
\item \textbf{Instance 4}: $X_1$ is $(1-(1-\varepsilon)(\min(\alpha_1,\alpha_2)-\varepsilon))/\varepsilon$ with probability $\varepsilon$ and $\min(\alpha_1,\alpha_2)-\varepsilon$ with probability $1-\varepsilon$; $X_2$ and $X_3$ are both $0$ deterministically. (This instance is used only if $\min(\alpha_1,\alpha_2)>0$.)
\end{itemize}
Observe that for each of the instances above, $\mathbb{E}[\max_i X_i]=\mathbb{E}[X_1]=1$. In order to prove an upper bound on the competitive ratio of ALG, we first consider the case where one of $\alpha_1$ and $\alpha_2$ is larger than $1$, and analyze the algorithm's performance on Instance 1.

\begin{lemma}
If $\alpha_1>1$ or $\alpha_2>1$, then the competitive ratio of ALG is no larger than $2/3$.
\end{lemma}

\begin{proof}
Consider Instance 1, where $X_1$ is $1$ deterministically, and $X_2$ and $X_3$ are both $0$ deterministically. Suppose $\alpha_1>1$. Then ALG misses $X_1$ with probability at least $1/3$, so the expectation of its profit is at most $2/3$. The same argument holds if $\alpha_2>1$.
\end{proof}

We therefore assume for the rest of this section that both $\alpha_1$ and $\alpha_2$ are at most $1$. In the next three lemmas, we analyze the algorithm's performance on Instances 2-4, each resulting in an upper bound on the algorithm's competitive ratio as a function of $\alpha_1$ and $\alpha_2$. Then we argue that the minimum of the three bounds is at most $11/15$ for any choice of $\alpha_1$ and $\alpha_2$.

\begin{lemma}\label{lem_1small}
The competitive ratio of ALG is no larger than $(1+2\alpha_1)/3$.
\end{lemma}

\begin{proof}
Consider Instance 2, where $X_1$ is $1$ deterministically, and $X_2$ and $X_3$ are both $\alpha_1$ deterministically. ALG necessarily accepts the first sample, because all the random variables are at least $\alpha_1$ with probability one. With probability $1/3$, $X_1$ appears first and ALG's profit is $1$, whereas with probability $2/3$, one of $X_2$ and $X_3$ comes first, resulting in the ALG's profit being $\alpha_1$. Thus, the expectation of ALG's profit is $(1+2\alpha_1)/3$. 
\end{proof}

\begin{lemma}\label{lem_1large2small}
If $\alpha_1>\alpha_2$, then the competitive ratio of ALG is no larger than $(2-\alpha_1+2\alpha_2)/3$.
\end{lemma}

\begin{proof}
Consider Instance 3, where $X_1$ is $(1-(1-\varepsilon)(\alpha_1-\varepsilon))/\varepsilon$ with probability $\varepsilon$ and $\alpha_1-\varepsilon$ with probability $1-\varepsilon$, whereas both $X_2$ and $X_3$ are $\alpha_2$ deterministically. Suppose $X_1$ appears first. If its value is realized to be $(1-(1-\varepsilon)(\alpha_1-\varepsilon))/\varepsilon$, then the algorithm picks it; otherwise the algorithm picks $\alpha_2$ in the next round. Suppose $X_1$ appears second. Then the first sample, which is necessarily $\alpha_2$, is rejected, and since $X_1\geq\alpha_2$ with probability one, the algorithm picks whatever value is realized for $X_1$. If $X_1$ appears last, then the algorithm rejects the first sample, which is $\alpha_2$, and accepts the second one, which is also $\alpha_2$. Thus, the expected profit of the algorithm is
\[\frac{1}{3}\cdot\left(\varepsilon\cdot\frac{1-(1-\varepsilon)(\alpha_1-\varepsilon)}{\varepsilon}+(1-\varepsilon)\cdot\alpha_2\right)+\frac{1}{3}\cdot\mathbb{E}[X_1]+\frac{1}{3}\cdot\alpha_2=\frac{1}{3}\cdot(1-(1-\varepsilon)(\alpha_1-\varepsilon-\alpha_2)+1+\alpha_2)\text{.}\]
As $\varepsilon\rightarrow0$, this approaches $(2-\alpha_1+2\alpha_2)/3$.
\end{proof}

\begin{lemma}\label{lem_2large}
If $\min(\alpha_1,\alpha_2)=\alpha>0$, then the competitive ratio of ALG is no larger than $(3-2\alpha)/3$.
\end{lemma}

\begin{proof}
Consider Instance 4, where $X_1$ is $(1-(1-\varepsilon)(\alpha-\varepsilon))/\varepsilon$ with probability $\varepsilon$ and $\alpha-\varepsilon$ with probability $1-\varepsilon$, whereas both $X_2$ and $X_3$ are $0$ deterministically. If $X_1$ appears first or second, then the algorithm picks it if and only if its realized value is $(1-(1-\varepsilon)(\alpha-\varepsilon))/\varepsilon$. If $X_1$ appears last, then the algorithm picks it irrespective of its realized value. Thus, the expected profit of the algorithm is
\[\frac{2}{3}\cdot\left(\varepsilon\cdot\frac{1-(1-\varepsilon)(\alpha-\varepsilon)}{\varepsilon}\right)+\frac{1}{3}\cdot\mathbb{E}[X_1]=\frac{1}{3}\cdot(2-2(1-\varepsilon)(\alpha-\varepsilon)+1)\text{.}\]
As $\varepsilon\rightarrow0$, this approaches $(3-2\alpha)/3$.
\end{proof}

The above lemmas allow us to prove Theorem \ref{thm_hardness} as follows.

\begin{proof}[Proof of Theorem \ref{thm_hardness}]
First, suppose $\alpha_1>\alpha_2>0$. Then by Lemmas \ref{lem_1small}, \ref{lem_1large2small}, and \ref{lem_2large}, the competitive ratio of the algorithm is at most
\[\frac{\min(1+2\alpha_1,2-\alpha_1+2\alpha_2,3-2\alpha_2)}{3}\leq\frac{1\times(1+2\alpha_1)+2\times(2-\alpha_1+2\alpha_2)+2\times(3-2\alpha_2)}{(1+2+2)\times3}=\frac{11}{15}\text{.}\]
Next, suppose $\alpha_1>\alpha_2=0$. Then by Lemmas \ref{lem_1small} and \ref{lem_1large2small}, the competitive ratio of the algorithm is at most
\[\frac{\min(1+2\alpha_1,2-\alpha_1)}{3}\leq\frac{1\times(1+2\alpha_1)+2\times(2-\alpha_1)}{(1+2)\times3}=\frac{5}{9}<\frac{11}{15}\text{.}\]
On the other hand, if $0<\alpha_1\leq\alpha_2$, then by Lemmas \ref{lem_1small} and \ref{lem_2large}, the competitive ratio of the algorithm is at most
\[\frac{\min(1+2\alpha_1,3-2\alpha_1)}{3}\leq\frac{1+2\alpha_1+3-2\alpha_1}{6}=\frac{2}{3}<\frac{11}{15}\text{.}\]
Finally, if $0=\alpha_1\leq\alpha_2$, then by Lemma \ref{lem_1small}, the competitive ratio of the algorithm is at most $1/3<11/15$.
\end{proof}

\section{Concluding Remarks}\label{sec_rem}

Our understanding of the Prophet Secretary problem is still limited, and there is a lot of scope for diving deeper. We showed that $1-1/e$ is not the correct competitive ratio for the Prophet Secretary problem. Under the natural restriction of deterministic distribution-insensitivity, we showed that no algorithm can have competitive ratio larger than $11/15$.  We conjecture that none of the bounds known bounds for the Prophet Secretary problem is tight. As this is a fundamental problem, finding the right competitive ratio is an important question which is still wide open.

\section*{Acknowledgments}

The authors thank Amos Fiat for his insightful involvement in the discussions. The second author thanks Matt Weinberg for introducing him to optimal stopping theory.


\bibliographystyle{plain}
\bibliography{references}


\appendix

\section{Proofs omitted from Section \ref{sec_alg}}

One claim which we will use repeatedly in this section is the following.

\begin{lemma}\label{lem_taylor}
For $x\geq 0$, $0\leq 1-e^{-x}\leq x$ and $0\leq 1-(1+x)e^{-x}\leq x^2$.
\end{lemma}

\begin{proof}
The lower bound on $1-e^{-x}$ is obvious. The lower bound on $1-(1+x)e^{-x}$ follows easily from the fact that $1+x\leq e^x$. Since $1-x\leq e^{-x}$, we have $1-e^{-x}\leq x$. Multiplying this by $1+x$, which is positive, we get $1+x-(1+x)e^{-x}\leq x^2+x$, that is, $1-(1+x)e^{-x}\leq x^2$, as required.
\end{proof}

\subsection{Proofs omitted from the analysis of Case 1}\label{app_1}

Let $\tau=\lfloor n(1+\ln(1-a))+1\rfloor$. Our next lemma states the asymptotic behavior of certain error terms which appear in the proof of Theorem \ref{thm_case1}.

\begin{lemma}\label{lem_asymptotic1}
Suppose $\alpha_1\ldots,\alpha_n$ are defined as in Case 1 of the algorithm. Then there exists absolute positive constants, $\gamma_1^>$, $\gamma_1^=$, $\gamma_1^<$, such that
\begin{enumerate}
\item For $k>\tau$, $\frac{1-\alpha_k}{n}-\alpha_k+\alpha_{k+1}\geq-\frac{\gamma_1^>}{n^2}$.
\item For $k=\tau$, $\frac{1-c\alpha_k}{n}-\alpha_k+\alpha_{k+1}\geq-\frac{\gamma_1^=}{n}$.
\item For $k<\tau$, $\frac{1-c\alpha_k}{n}-\alpha_k+\alpha_{k+1}\geq-\frac{\gamma_1^<}{n^2}$.
\end{enumerate} 
\end{lemma}

\begin{proof}
For $k>\tau$, we have
\[\frac{1-\alpha_k}{n}-\alpha_k+\alpha_{k+1}=\frac{e^{(k-1)/n-1}}{n}+e^{(k-1)/n-1}-e^{k/n-1}=-e^{k/n-1}\left(-e^{-1/n}\left(\frac{1}{n}+1\right)+1\right)\text{.}\]
Using Lemma \ref{lem_taylor} for the expression in the parenthesis and using $k\leq n$, we have,
\[\frac{1-\alpha_k}{n}-\alpha_k+\alpha_{k+1}\geq-\frac{e^{k/n-1}}{n^2}\geq-\frac{1}{n^2}=-\frac{\gamma_1^>}{n^2}\]
where $\gamma_1^>=1$.

For $k=\tau$, we have
\begin{eqnarray*}
\frac{1-c\alpha_k}{n}-\alpha_k+\alpha_{k+1} & = & \frac{(1-ca)}{n}\cdot\left(\frac{e^{(k-1)/n-1}}{1-a}\right)^c-\frac{1}{c}+\left(\frac{1}{c}-a\right)\left(\frac{e^{(k-1)/n-1}}{1-a}\right)^c+1-e^{k/n-1}\\
 & = & \left(\frac{1}{c}-a\right)\left(\frac{e^{(k-1)/n-1}}{1-a}\right)^c\left(\frac{c}{n}+1\right)-\frac{1}{c}+1-e^{k/n-1}\text{.}
\end{eqnarray*}
Since $n(1+\ln(1-a))<k=\tau\leq n(1+\ln(1-a))+1$, we have $e^{(k-1)/n-1}=e^{k/n-1/n-1}\geq e^{1+\ln(1-a)}\cdot e^{-1-1/n}=(1-a)\cdot e^{-1/n}$, and $e^{k/n-1}\leq e^{\ln(1-a)+1/n}=(1-a)e^{1/n}$. Substituting these bounds above, we get
\begin{eqnarray*}
\frac{1-c\alpha_k}{n}-\alpha_k+\alpha_{k+1} & \geq & \left(\frac{1}{c}-a\right)e^{-c/n}\left(\frac{c}{n}+1\right)-\frac{1}{c}+1-(1-a)e^{1/n}\\
 & = & \frac{e^{-c/n}}{c}\left(\frac{c}{n}+1\right)-ae^{-c/n}\left(\frac{c}{n}+1\right)-\frac{1}{c}+1-(1-a)e^{1/n}\\
 & = & -\frac{1}{c}\left(-e^{-c/n}\left(\frac{c}{n}+1\right)+1\right)+a\left(1-e^{-c/n}\left(\frac{c}{n}+1\right)\right)-(1-a)(e^{1/n}-1)\text{.}
\end{eqnarray*}
Using Lemma \ref{lem_taylor} and the fact that $0\leq a\leq1$, we have,
\[\frac{1-c\alpha_k}{n}-\alpha_k+\alpha_{k+1}\geq-\frac{1}{c}\cdot\frac{c^2}{n^2}-\frac{(1-a)e^{1/n}}{n}\geq-\frac{c}{n}-\frac{(1-a)e}{n}=-\frac{c+(1-a)e}{n}=-\frac{\gamma_1^=}{n}\]
where $\gamma_1^==c+(1-a)e$.

For $k<\tau$, we have
\begin{eqnarray*}
\frac{1-c\alpha_k}{n}-\alpha_k+\alpha_{k+1} & = & \frac{(1-ca)}{n}\cdot\left(\frac{e^{(k-1)/n-1}}{1-a}\right)^c+\left(\frac{1}{c}-a\right)\left(\frac{e^{(k-1)/n-1}}{1-a}\right)^c-\left(\frac{1}{c}-a\right)\left(\frac{e^{k/n-1}}{1-a}\right)^c\\
 & = & -\left(\frac{1}{c}-a\right)\left(\frac{e^{k/n-1}}{1-a}\right)^c\left(-e^{-c/n}\left(\frac{c}{n}+1\right)+1\right)\text{.}
\end{eqnarray*}
Using Lemma \ref{lem_taylor} and using the facts $k\leq n$ and $0\leq a\leq 1\leq1/c$, we have,
\[\frac{1-c\alpha_k}{n}-\alpha_k+\alpha_{k+1}\geq-\left(\frac{1}{c}-a\right)\cdot\frac{1}{(1-a)^c}\cdot\frac{c^2}{n^2}=-\frac{\gamma_1^<}{n^2}\]
where $\gamma_1^<=(1/c-a)\cdot(1-a)^{-c}$.
\end{proof}

\begin{proof}[Proof of Theorem \ref{thm_case1}]
We have
\begin{eqnarray*}
\sum_{k=1}^n\int_{\alpha_k}^{\infty}\Pr[z_k\geq x]dx & = & \sum_{k=1}^\tau\int_{\alpha_k}^{\infty}\Pr[z_k\geq x]dx+\sum_{k=\tau+1}^n\int_{\alpha_k}^{\infty}\Pr[z_k\geq x]dx\\
 & \geq & \sum_{k=1}^\tau\theta(k)\cdot\frac{1-c\alpha_k}{n}+\sum_{k=\tau+1}^n\theta(k)\cdot\frac{1-\alpha_k}{n}\text{.}
\end{eqnarray*}
Here we used Lemma \ref{lem_case1} for the first term and Proposition \ref{prop_infty} for the second term. Substituting this and the bound of Proposition \ref{prop_0} in Equation (\ref{eqn_alg_esa}), we get
\[\alg\geq\alpha_1+\sum_{k=1}^{\tau}\theta(k)\cdot\left(\frac{1-c\alpha_k}{n}-\alpha_k+\alpha_{k+1}\right)+\sum_{k=\tau+1}^n\theta(k)\cdot\left(\frac{1-\alpha_k}{n}-\alpha_k+\alpha_{k+1}\right)\text{.}\]
Using Lemma \ref{lem_asymptotic1} and the fact that $\theta(k)\leq1$, we have,
\begin{eqnarray*}
\alg & \geq & \alpha_1-\sum_{k=1}^{\tau-1}\theta(k)\cdot\frac{\gamma_1^<}{n^2}-\theta(\tau)\cdot\frac{\gamma_1^=}{n}-\sum_{k=\tau+1}^n\theta(k)\cdot\frac{\gamma_1^>}{n^2}\\
 & \geq & \alpha_1-\frac{\gamma_1^<+\gamma_1^=+\gamma_1^>}{n}\\
 & = & \frac{1}{c}-\left(\frac{1}{c}-a\right)(e(1-a))^{-c}-\frac{\gamma_1}{n}
\end{eqnarray*}
where $\gamma_1=\gamma_1^<+\gamma_1^=+\gamma_1^>$.
\end{proof}

\subsection{Proofs omitted from the analysis of Case 2}\label{app_2}

Let $\tau=\left\lfloor n\left(1-\frac{1}{b}\ln\frac{b+d-bd}{d}\right)+1\right\rfloor$. Our next lemma states the asymptotic behavior of certain error terms which appear in the proof of Theorem \ref{thm_case2}.

\begin{lemma}\label{lem_asymptotic2}
Suppose $\alpha_1\ldots,\alpha_n$ are defined as in Case 2 of the algorithm. Then there exists absolute positive constants, $\gamma_2^>$, $\gamma_2^=$, $\gamma_2^<$, such that
\begin{enumerate}
\item For $k<\tau$, $\frac{1-\alpha_k}{n}-\alpha_k+\alpha_{k+1}\geq-\frac{\gamma_2^<}{n^2}$.
\item For $k=\tau$, $\frac{1-\alpha_k}{n}-\alpha_k+\alpha_{k+1}\geq-\frac{\gamma_2^=}{n}$.
\item For $k>\tau$, $\frac{b+d-bd-b\alpha_k}{n}-\alpha_k+\alpha_{k+1}\geq-\frac{\gamma_2^>}{n^2}$.
\end{enumerate} 
\end{lemma}

\begin{proof}
For $k<\tau$, we have
\begin{eqnarray*}
\frac{1-\alpha_k}{n}-\alpha_k+\alpha_{k+1} & = & d\cdot\left(\frac{b+d-bd}{d}\right)^{1/b}\cdot\left(\frac{e^{(k-1)/n-1}}{n}+e^{(k-1)/n-1}-e^{k/n-1}\right)\\
 & = & -d\cdot\left(\frac{b+d-bd}{d}\right)^{1/b}\cdot e^{k/n-1}\cdot\left(-e^{-1/n}\left(\frac{1}{n}+1\right)+1\right)\text{.}
\end{eqnarray*}
Using Lemma \ref{lem_taylor} and the fact that $k\leq n$, we have,
\[\frac{1-\alpha_k}{n}-\alpha_k+\alpha_{k+1}\geq-d\cdot\left(\frac{b+d-bd}{d}\right)^{1/b}\cdot\frac{1}{n^2}=-\frac{\gamma_2^<}{n^2}\]
where $\gamma_2^<=d((b+d+bd)/d)^{1/b}$.

For $k=\tau$, we have
\[1-\frac{1}{b}\ln\frac{b+d-bd}{d}<\frac{k}{n}\leq1-\frac{1}{b}\ln\frac{b+d-bd}{d}+\frac{1}{n}\text{.}\]
Thus,
\begin{eqnarray}
\frac{1-\alpha_k}{n}-\alpha_k & = & \frac{d}{n}\left(\frac{b+d-bd}{d}\right)^{1/b}e^{(k-1)/n-1}-1+d\left(\frac{b+d-bd}{d}\right)^{1/b}e^{(k-1)/n-1}\nonumber\\
 & = & d\left(\frac{b+d-bd}{d}\right)^{1/b}e^{(k-1)/n-1}\left(\frac{1}{n}+1\right)-1\nonumber\\
 & \geq & de^{-1/n}\left(\frac{1}{n}+1\right)-1\label{eqn_asym1}
\end{eqnarray}
where we used the lower bound on $k/n$ in the inequality. We also have,
\begin{eqnarray}
\alpha_{k+1} & = & (1-e^{b(k/n-1)})\frac{b+d-bd}{b}\nonumber\\
 & \geq & \left(1-\frac{de^{b/n}}{b+d-bd}\right)\frac{b+d-bd}{b}\nonumber\\
 & = & \frac{b+d-bd}{b}-\frac{de^{b/n}}{b}\nonumber\\
 & = & 1+\frac{d}{b}-d-\frac{de^{b/n}}{b}\label{eqn_asym2}
\end{eqnarray}
where we used the upper bound on $k/n$ in the inequality. Adding (\ref{eqn_asym1}) and (\ref{eqn_asym2}),
\begin{eqnarray*}
\frac{1-\alpha_k}{n}-\alpha_k+\alpha_{k+1} & \geq & de^{-1/n}\left(\frac{1}{n}+1\right)+\frac{d}{b}-d-\frac{de^{b/n}}{b}\\
 & \geq & -d\left(-e^{-1/n}\left(\frac{1}{n}+1\right)+1\right)-\frac{de^{b/n}}{b}(1-e^{-b/n})\text{.}
\end{eqnarray*}
Using Lemma \ref{lem_taylor}, we have,
\[\frac{1-\alpha_k}{n}-\alpha_k+\alpha_{k+1}\geq-\frac{d}{n^2}-\frac{de^{b/n}}{b}\cdot\frac{b}{n}\geq-\frac{d}{n}-\frac{de^b}{n}=-\frac{d(1+e^b)}{n}=-\frac{\gamma_2^=}{n}\]
where $\gamma_2^==d(1+e^b)$.

For $k>\tau$, we have
\[\frac{b+d-bd-b\alpha_k}{n}=\frac{(b+d-bd)\cdot e^{b((k-1)/n-1)}}{n}\text{,}\]
\[\alpha_{k+1}-\alpha_k=(e^{b((k-1)/n-1)}-e^{b(k/n-1)})\cdot\frac{b+d-bd}{b}=\frac{b+d-bd}{b}\cdot e^{b((k-1)/n-1)}\cdot(1-e^{b/n})\text{.}\]
Adding the above two equations, we get,
\begin{eqnarray*}
\frac{b+d-bd-b\alpha_k}{n}-\alpha_k+\alpha_{k+1} & = & \frac{b+d-bd}{b}\cdot e^{b((k-1)/n-1)}\cdot\left(\frac{b}{n}+1-e^{b/n}\right)\\
  & = & -\frac{b+d-bd}{b}\cdot e^{b(k/n-1)}\cdot\left(-e^{b/n}\left(\frac{b}{n}+1\right)+1\right)\text{.}
\end{eqnarray*}
Using Lemma \ref{lem_taylor} and the fact that $k\leq n$, we have,
\[\frac{b+d-bd-b\alpha_k}{n}-\alpha_k+\alpha_{k+1}\geq-\frac{b+d-bd}{b}\cdot\frac{b^2}{n^2}=-\frac{b(b+d-bd)}{n^2}=-\frac{\gamma_2^>}{n^2}\]
where $\gamma_2^>=b(b+d-bd)$.
\end{proof}

\begin{proof}[Proof of Theorem \ref{thm_case2}]
We have
\begin{eqnarray*}
\sum_{k=1}^n\int_{\alpha_k}^{\infty}\Pr[z_k\geq x]dx & = & \sum_{k=1}^\tau\int_{\alpha_k}^{\infty}\Pr[z_k\geq x]dx+\sum_{k=\tau+1}^n\int_{\alpha_k}^{\infty}\Pr[z_k\geq x]dx\\
 & \geq & \sum_{k=1}^\tau\theta(k)\cdot\frac{1-\alpha_k}{n}+\sum_{k=\tau+1}^n\theta(k)\cdot\frac{b+d-bd-b\alpha_k}{n}\text{.}
\end{eqnarray*}
Here we used Proposition \ref{prop_infty} for the first term and Lemma \ref{lem_case2} for the second term. Substituting this and the bound of Proposition \ref{prop_0} in Equation (\ref{eqn_alg_esa}), we get
\[\alg\geq\alpha_1+\sum_{k=1}^{\tau}\theta(k)\cdot\left(\frac{1-\alpha_k}{n}-\alpha_k+\alpha_{k+1}\right)+\sum_{k=\tau+1}^n\theta(k)\cdot\left(\frac{b+d-bd-b\alpha_k}{n}-\alpha_k+\alpha_{k+1}\right)\text{.}\]
Using Lemma \ref{lem_asymptotic2} and the fact that $\theta(k)\leq1$, we have,
\begin{eqnarray*}
\alg & \geq & \alpha_1-\sum_{k=1}^{\tau-1}\theta(k)\cdot\frac{\gamma_2^<}{n^2}-\theta(\tau)\cdot\frac{\gamma_2^=}{n}-\sum_{k=\tau+1}^n\theta(k)\cdot\frac{\gamma_2^>}{n^2}\\
 & \geq & \alpha_1-\frac{\gamma_2^<+\gamma_2^=+\gamma_2^>}{n}\\
 & = & 1-\frac{d}{e}\left(\frac{b+d-bd}{d}\right)^{1/b}-\frac{\gamma_2}{n}
\end{eqnarray*}
where $\gamma_2=\gamma_2^<+\gamma_2^=+\gamma_2^>$.
\end{proof}




\section{The IID Prophet Inequality}\label{app_iid}

Hill and Kertz \cite{HillK_AP82} analyzed the optimal dynamic programming algorithm for the IID Prophet Inequality problem, and characterized its competitive ratio implicitly as follows.

\begin{definition}
For each $n>1$, the function $\phi_n:[0,\infty)\times[0,\infty)\longrightarrow\mathbb{R}$ is defined as
\[\phi_n(w,x)=\frac{n}{n-1}w^{(n-1)/n}+\frac{x}{n-1}\text{.}\]
For each $n>1$ and $j=0,\ldots,n-1$, the function $\eta_{j,n}:[0,\infty)\longrightarrow\mathbb{R}$ is defined recursively as
\[\eta_{0,n}(\alpha)=\phi_n(0,\alpha)\text{,}\]
\[\eta_{j,n}(\alpha)=\phi_n(\eta_{j-1,n}(\alpha),\alpha)\text{.}\]
\end{definition}

With the above definition, Hill and Kertz proved the following claim on the competitive ratio for the IID Prophet Inequality.

\begin{claim}[Theorem A of Hill and Kertz \cite{HillK_AP82}]
There is a unique $\alpha_n>0$ for which $\eta_{n-1,n}(\alpha_n)=1$. Furthermore, the competitive ratio for the Prophet Inequality with $n$ IID random variables is \textbf{exactly} $c_{\text{iid}}(n)=(1+\alpha_n)^{-1}$.
\end{claim}

Hill and Kertz managed to determine the value of $c_{iid}(n)$ with a reasonable accuracy for $n\leq10000$, for instance, they found that $\alpha_{10000}\approx0.341$, implying $c_{\text{iid}}\leq c_{\text{iid}}(10000)\approx1/1.341<0.746$. However, for larger $n$, they only proved a (weak) lower bound of $1-1/e$ on $c_{iid}(n)$.

Recently, Correa et al.\ \cite{CorreaFHOV_EC17}, while analyzing posted price mechanisms for a randomly arranged sequence of customers, also showed that their techniques result in an algorithm for the IID Prophet Inequality. They gave the following implicitly defined bound on its competitive ratio.

\begin{claim}[Correa et al.\ \cite{CorreaFHOV_EC17}\footnote{Correa et al.\ do not state this as an explicit theorem. However, they mention essentially this in the section with the heading ``Bounding $\alpha_1$ through a recursion''.}]\label{claim_CFHOV}
There exists a unique sequence $x_0,\ldots,x_n$ such that $x_0=1$, $x_n=0$, and for every $i=1,\ldots,n-1$,
\[\frac{x_{i-1}^n}{n}-\frac{x_i^n}{n}=\frac{x_i^{n-1}}{n-1}-\frac{x_{i+1}^{n-1}}{n-1}\text{.}\]
Furthermore, the competitive ratio of Correa et al.'s algorithm for the Prophet Inequality with $n$ IID random variables is \textbf{at least} $c_{\text{cfhov}}(n)=(n(1-x_1^{n-1}))^{-1}$.
\end{claim}

Correa et al.\ also proved $c_{\text{cfhov}}=\inf_nc_{\text{cfhov}}(n)>0.745$. Together with the upper bound on $c_{\text{iid}}$, we immediately have $0.745<c_{\text{cfhov}}\leq c_{\text{iid}}<0.746$, which means that $c_{\text{cfhov}}$ and $c_{\text{iid}}$ are already very close. Surprisingly, they are, in fact, equal.

\begin{claim}
For all $n>1$, $c_{\text{cfhov}}(n)=c_{\text{iid}}(n)$. Hence, $c_{\text{cfhov}}=c_{\text{iid}}$.
\end{claim}

\begin{proof}
An equivalent way of defining the functions $\eta_{j,n}$ is the following. $\eta_{-1,n}(\alpha)=0$ and $\eta_{j,n}(\alpha)=\phi_n(\eta_{j-1,n}(\alpha),\alpha)$ for $j=0,\ldots,n-1$. For any $j=0,\ldots,n-2$ and any $\alpha$, we have,
\begin{eqnarray}
\eta_{j,n}(\alpha) & = & \frac{n}{n-1}(\eta_{j-1,n}(\alpha))^{(n-1)/n}+\frac{\alpha}{n-1}\text{,}\label{eqn_eta_rec1}\\
\eta_{j+1,n}(\alpha) & = & \frac{n}{n-1}(\eta_{j,n}(\alpha))^{(n-1)/n}+\frac{\alpha}{n-1}\text{.}\nonumber
\end{eqnarray}
Subtracting, we get, for any $j=0,\ldots,n-2$ and any $\alpha$,
\begin{equation}\label{eqn_eta_rec2}
\frac{\eta_{j+1,n}(\alpha)}{n}-\frac{\eta_{j,n}(\alpha)}{n}=\frac{(\eta_{j,n}(\alpha))^{(n-1)/n}}{n-1}-\frac{(\eta_{j-1,n}(\alpha))^{(n-1)/n}}{n-1}\text{.}
\end{equation}

Let $y_i=(\eta_{n-1-i,n}(\alpha_n))^{1/n}$ for $i=0,\ldots,n-1$, and $y_n=0$. Then $y_0=(\eta_{n-1,n}(\alpha_n))^{1/n}=1$, because $\eta_{n-1,n}(\alpha_n)=1$, by definition of $\alpha_n$. Also, by Equation (\ref{eqn_eta_rec2}), we have
\[\frac{y_{i-1}^n}{n}-\frac{y_i^n}{n}=\frac{y_i^{n-1}}{n-1}-\frac{y_{i+1}^{n-1}}{n-1}\text{.}\]
Thus, by Claim \ref{claim_CFHOV}, $x_i=y_i$ for all $i$. Also, substituting $j=n-1$ and $\alpha=\alpha_n$ in Equation (\ref{eqn_eta_rec1}), we get,
\[1=\eta_{n-1,n}(\alpha_n)=\frac{n}{n-1}(\eta_{n-2,n}(\alpha_n))^{(n-1)/n}+\frac{\alpha_n}{n-1}\text{.}\]
Thus,
\[y_1^{n-1}=(\eta_{n-2,n}(\alpha_n))^{(n-1)/n}=\frac{n-1-\alpha_n}{n}\text{,}\]
that is,
\[1-y_1^{n-1}=\frac{1+\alpha_n}{n}\text{.}\]
Therefore,
\[c_{\text{cfhov}}(n)=\frac{1}{n(1-x_1^{n-1})}=\frac{1}{n(1-y_1^{n-1})}=\frac{1}{1+\alpha_n}=c_{\text{iid}}(n)\text{.}\]
\end{proof}

This means that Correa et al.'s algorithm also achieves the optimal competitive ratio for IID Prophet Inequality. However, while the optimal dynamic programming algorithm is guaranteed to achieve the largest possible profit on every input, it is unclear whether Correa et al.'s algorithm has this property.

\end{document}